\documentclass{article}

\usepackage{amssymb,amsmath, amsthm}
\usepackage[T1]{fontenc}
\usepackage{comment}
\usepackage[utf8]{inputenc}
\usepackage[usenames,dvipsnames]{xcolor}
\usepackage{tikz}
\usetikzlibrary{shapes, arrows, chains}
\usetikzlibrary{matrix}
\usepackage{graphicx}
\usepackage{wrapfig}
\usepackage{framed}

\newtheorem{theorem}{Theorem}
\newtheorem{lemma}{Lemma}

\newtheorem{example}{Example}
\newtheorem{definition}{Definition}

\usepackage{algorithm}
\usepackage{algpseudocodex}
\usepackage{todonotes}
\usepackage{xspace}
\usepackage{hyperref}
\usepackage{xurl}
\usepackage{stmaryrd}

\newcommand{\el}{\ensuremath{\mathcal{EL}}\xspace}
\newcommand{\flo}{\ensuremath{\mathcal{FL}_0}\xspace}

\newcommand{\wrt}{w.r.t.\ }
\newcommand{\ie}{i.e.\ }

\newcommand{\var}{\ensuremath{\mathbf{Var}}\xspace}
\newcommand{\roles}{\ensuremath{\mathbf{R}}\xspace}
\newcommand{\names}{\ensuremath{\mathbf{N}}\xspace}
\newcommand{\const}{\ensuremath{\mathbf{C}}\xspace}

\renewcommand{\S}{\ensuremath{\mathcal{S}}}
\date{\today}
\author{Barbara Morawska, Dariusz Marzec, Sławomir Kost, Michał Henne\\
	\texttt{\small\{barbara.morawska,  dariusz.marzec, slawomir.kost, michal.henne\}@uni.opole.pl}
	\thanks{\small	
			\begin{wrapfigure}[4]{r}{0.15\textwidth}
				\vspace{-5mm}
				\hspace{-3mm}
				\includegraphics[width=2cm]{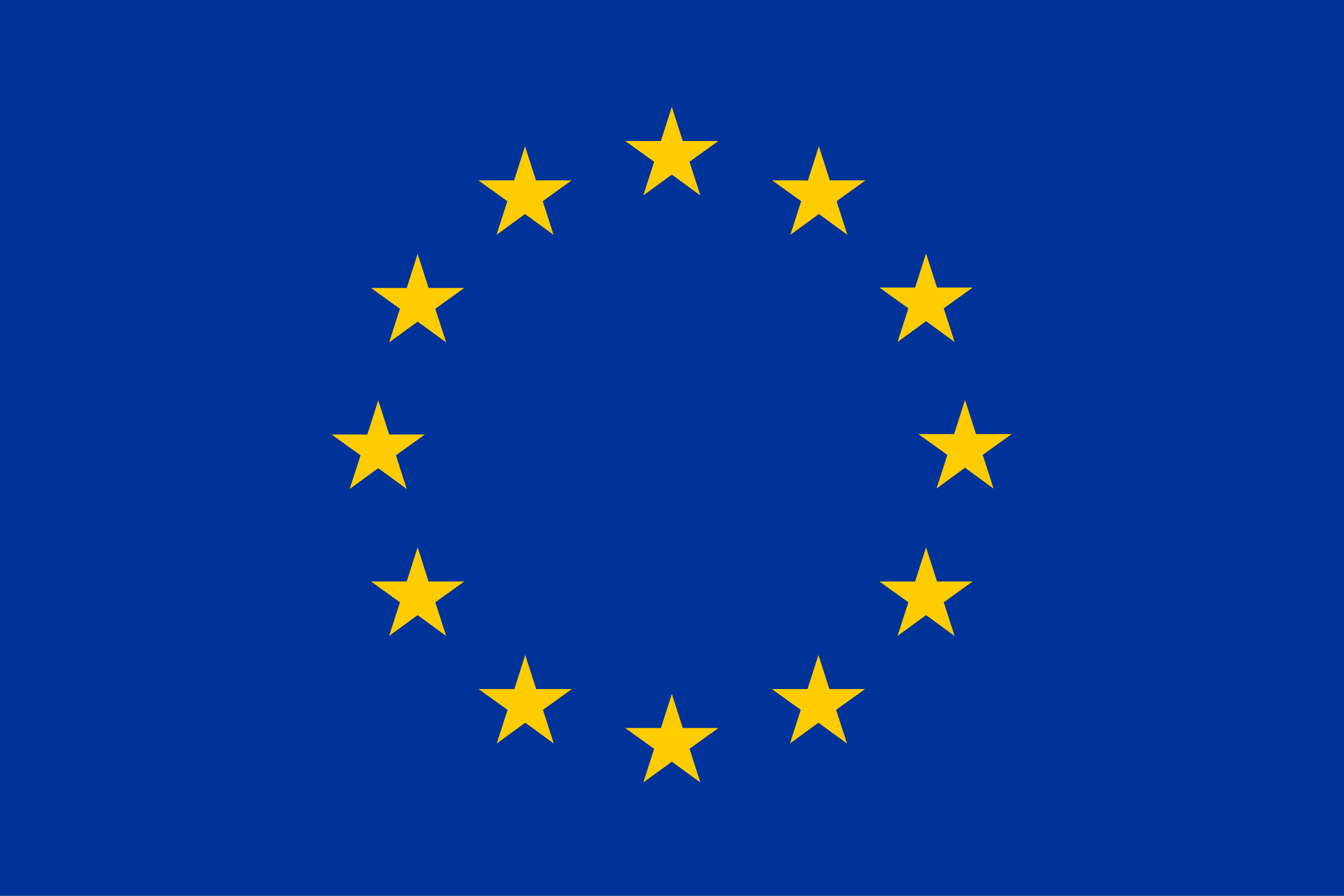}
			\end{wrapfigure}
			This research is part of the project No 2022/47/P/ST6/03196 within the POLONEZ BIS
			programme co-funded by the National Science Centre and the European Union’s Horizon 2020
			research and innovation programme under the Marie Skłodowska-Curie grant agreement
			No. 945339. For the purpose of Open Access, the author has applied a CC-BY public copyright
			licence to any Author Accepted Manuscript (AAM) version arising from this submission.		
			}}
\title{FILO -- automated unification in \flo\\
 (extended version)}
\begin{document}
	\maketitle
	\begin{abstract}
	FILO is a java application that decides unifiability for a unification problem formulated in 
	the description logic \flo. If the problem is unifiable, it presents a user with
	an example of a solution. FILO joins a family of similar applications like UEL\cite{Baader2012} solving unification problems in
	the description logic \el, \flo{wer} \cite{Baader2022}, a subsumption decider for \flo with TBox, CEL \cite{Baader2006} and JCEL \cite{Mendez2012} subsumption deciders
	for \el with TBox, and others. These systems play an important role in various knowledge representation reasoning problems.
	
\end{abstract}
	
	\section{Introduction}
	Description Logics (DLs) are a formalism used to represent knowledge in a given domain.
	Usually, this knowledge is stored in the form of an ontology, a large set of definitions of concepts.
	DLs describe, in a formal way, how complex concepts can be constructed from simpler ones, \ie
	concept names and role names (binary relations). DLs differ between themselves by providing different sets of constructors that can be used to construct complex concepts.
	The description logic \flo is a member of a family of small description logics with restricted expressive power. It provides only \emph{intersection (conjunction)} constructor, $\top$ constructor, and value restriction of the form
	$\forall r.C$, where $C$ is a concept and $r$ is a role name. A value restriction for a role $r$ expresses the restriction on the elements related to a given element with the role $r$ to the objects belonging to a concept $C$.
		
	For example in the notation of \flo we can define the following two concepts:
	
	\begin{itemize}
		\item  \texttt{Student $\sqcap\,\, \forall$attends.(PGCourse $\sqcap\,\, \forall$given\_by.(Professor $\sqcap \forall$belonging\_to. CSFaculty)) },
		\item  \texttt{PostGraduateStudent $\sqcap\,\, \forall$attends.CSCourse}.
	\end{itemize}
	
	These concepts are not equal. Unification asks for a substitution, or a definition of some of the component concept names, which
	may make the concepts equivalent. For example if we treat \texttt{PostGraduateStudent} and \texttt{CSCourse} as \emph{variables},
	then we can say that these two concepts are equivalent under the definitions:
	
	\begin{itemize}
		\item \texttt{PostGraduateStudent $:=$ Student $\sqcap$ $\forall${attends}.PGCourse} 
		\item \texttt{CSCourse $:=$ PGCourse $\sqcap\, \forall$given\_by.(Professor $\sqcap \forall$belonging\_to.CSFaculty)}.
	\end{itemize}

	The unification problem in DLs was, in fact, first defined for our logic \flo. It was shown in \cite{BaNa-JSC01} to be
	an Exptime-complete problem. The algorithm presented there reduced the problem to the emptiness problem for
	a root-to-frontier automaton on trees. Since the automaton is of size exponential in the size of
	a unification problem, this showed the exponential upper bound. For some reasons, this algorithm
	was never implemented. One of them was prehaps that the focus of researchers turned to another
	small DL, namelly to \el. The small description logic \el does not have a value restriction constructor,
	but provides instead an existential restriction: $\exists r.C$, which expresses an existential requirement
	that an object be related by a relation $r$ to another object which satisfies the concept $C$.
	Unification of concepts formulated in \el was solved in \cite{Baader2010} and 
	shown to be an NP-complete problem. There are at least two algorithms, one a goal-directed based on inference rules\cite{Baader2012c} and the other based on a SAT-reduction \cite{Baader2012a}, that establish the upper bound.
	The algorithms for solving unification in \el were automatized as a Java application UEL, \cite{Baader2012b}.
	
	In \cite{Borgwardt2012} we have revisited the unification algorithm for \flo and 
	developed a new one, based on reduction to a problem of finding a special kind of models for 
	a set of anti-Horn clauses, which seemed a bit analogous to the SAT-reduction on the side of \el.
	In the present paper we come back to that algorithm, and revise it again, so that it can be 
	implemented. Our implementation called FILO works on unification problems in form of ontology files,
	where variables are marked with the suffix \emph{var}. Basically FILO is a decision procedure,
	its purpose is not to compute a  unifier or a set of unifiers. Nevertheless if it detects that
	an input problem is unifiable, it will output a unifier, that can be extracted from its computations.
	
	\section{The description logic \flo}
	As mentioned in the introduction, the description logic \flo deals with concepts constructed
	from a countable set of concept names (unary predicates) \names and a countable set of role names (binary predicates, binary relations) \roles.
	Concepts are generated by the following grammar.
	
	\[ C ::= A \mid C \sqcap C \mid \forall r.C \mid \top \]
	where $A$ is an arbitrary element of \names and $r$ is an element of \roles.
	
	Concept names are interpreted as subsets of a non-empty domain and role names as binary relations.
	Hence an \flo-interpretation $I$ is a pair $(\Delta^I, \cdot^I)$, where $\Delta^I$ is a non-empty domain
	and $\cdot^I$ is an interpreting function. Since we are concerned with only \flo-interpretations, we will talk about just interpretations, omitting the prefix $\flo$.
	 The interpretation for concept names ($A^I \subseteq \Delta^I$) and role names ($r^I \subseteq \Delta^I \times \Delta^I$) is extended
	to all \flo-concepts in the following way.
	
	\begin{itemize}
		\item $(C_1 \sqcap C_2)^I = C_1^I \cap C_2^I$
		\item $(\forall r.C)^I = \{e \in \Delta^I \mid \forall_{d \in \Delta^I} ((e,d) \in r^I \implies d \in C^I)\}$
		\item $\top^I = \Delta^I$
	\end{itemize}

Based on this semantics we define the subsumption relation between concepts.
$C \sqsubseteq D$ iff for every interpretation $I$, $C^I \subseteq D^I$.
The equivalence of concepts is then defined as the subsumption in both directions.
$C \equiv D$ iff $C\sqsubseteq D$ and $D \sqsubseteq C$.	
	
With respect to the equivalence, one can easily notice that the intersection of concepts is associative,
commutative and has $\top$ as its unit. The value restrictions behave as homomorphisms:
$\forall r.(C_1 \sqcap C_2) \equiv \forall r.C_1 \sqcap \forall r.C_2$.

Due to these properties each concept is equivalent to a conjunction of concepts of the form:
$\forall r_1.\forall r_2. \dots \forall r_n. A$, where $r_1, \dots, r_n$ are not necessarily different role names and $A$ is a concept name. For brevity we write $\forall v.A$, where $v$ is a word over \roles, 
and call such a concept \emph{a particle}.

Hence, each concept is equivalent to a conjunction of particles, and an empty conjunction is considered to be $\top$. We call such a concept: a conjunction of particles, a concept \emph{in normal form}, and
since conjunction is associative, commutative and idempotent, we identify it with a set.
Hence each concept in normal form is a set of particles, and the empty set is $\top$.

Below we will assume that each concept is in normal form.

We can also assume that it is reduced, \ie no one of the following rules applies to it.
\begin{enumerate}
	\item $C \sqcap \top \leadsto C$, $\top \sqcap C \leadsto C$,
	\item $\forall v.\top \leadsto \top$
\end{enumerate}

In \flo subsumption between two concepts $C$ and $D$ can be decided in polynomial time.
\begin{lemma}\label{lemma:subsumption} Let  $C=P_1 \sqcap \cdots\sqcap P_m$ and $D=P'_1 \sqcap \cdots \sqcap P'_n$.
	\begin{enumerate}
\item 	$C \sqsubseteq D $ iff for every  $1 \le i \le n$,  $P_1 \sqcap \cdots\sqcap P_m \sqsubseteq P'_i$
\item for every $1 \le i \le n$,  $P_1 \sqcap \cdots\sqcap P_m \sqsubseteq P'_i$ iff 
$P'_i \in \{P_1, \cdots, P_m\}$.
\end{enumerate}
\end{lemma}
\begin{proof}
	Both statements follow from the properties of subsumption that establishes a partial order with respect to the sets of particles.
\end{proof}

The next lemma justifies dividing a decision of a subsumption relation between two concepts into a set of decisions about smaller subsumptions that are restricted \wrt different constants.

\begin{lemma}\label{lemma:one-constant}
	 Let  $C=P_1 \sqcap \cdots\sqcap P_m$ and $D=P'_1 \sqcap \cdots \sqcap P'_n$. 
	 Then 	$C \sqsubseteq D $ iff for every $A \in \names$,
	 $C_A \sqsubseteq D_A$, where $C_A, D_A$ are obtained from  $C$ and $D$ by deleting all particles
	 of the form $\forall v.B$, with $B \not=A$.
\end{lemma}
\begin{proof}
	Assume that $C \sqsubseteq D$.
	Let $A$ occur in $D$. Hence $D_A = \{P'_i \mid P'_i = \forall v_i.A \}$.
	By Lemma~\ref{lemma:subsumption}, $\forall v_i.A \in C$ but it is also true that
	$\forall v_i.A \in C_A$.
	If $A$ does not occur in $D$, then $D_A = \top$ and thus $C_A \sqsubseteq D_A$.
	
	Now assume that for every $A \in \names$, $C_A \sqsubseteq D_A$.
	Then by Lemma~\ref{lemma:subsumption} $\bigsqcap \{C_A \mid A \in \names\} \sqsubseteq \bigsqcap \{D_A \mid A \in \names\} $.
	And for reduced concepts, $\bigsqcap \{C_A \mid A \in \names\} = C$ and 
	$\bigsqcap \{D_A \mid A \in \names\}= D $, hence $C \sqsubseteq D$.
\end{proof}

	\section{Unification problem}
	
	In order to define a unification problem in \flo, we have to decide which of the concept names are variables and which are constants. The variables may be substituted by concepts and constants cannot be substituted. Hence we divide the set of concept names \names into two disjoint sets 
	\const which will be called \emph{constants} and \var, \emph{variables}.
	
	The unification problem is then defined by its input and output as follows.
	
	\textbf{Input:}
	$\Gamma = \{C_1 \sqsubseteq^? D_1, \dots, C_n \sqsubseteq^? D_n\}$,
	where $C_1\ldots C_n, D_1\ldots  D_n$ are \flo-concepts constructed over constants and variables.
	We call $C \sqsubseteq^? D \in \Gamma$ a \emph{goal subsumption}.

\textbf{Output:} "true'' if there is a substitution $\gamma$ such that\\
$\gamma(C_1) \sqsubseteq \gamma( D_1), \dots, \gamma(C_n) \sqsubseteq \gamma(D_n)$.
The substitution $\gamma$ is called \emph{a unifier} or \emph{a solution} of $\Gamma$.

Due to Lemma~\ref{lemma:subsumption}, we can assume that for each goal subsumption, $C \sqsubseteq^? D$, $D$ is a particle.
	
\section{Solver}	
	
FILO is an application written in Java using OWL API and Maven for a dependency management.
As for now it is a standalone application (FILO.jar). The application can be easily opened by double-clicking the file (Windows) or running \texttt{java - jar Filo.jar} in the command line (Linux).
 The compiled file is available at \href{https://unifdl.cs.uni.opole.pl/unificator-app-for-the-description-logic-fl_0}{\url{https://unifdl.cs.uni.opole.pl/unificator-app-for-the-description-logic-fl_0/ }}
  and the source files are available at the public GitHub repository \href{https://github.com/barbmor/FILO}{\url{https://github.com/barbmor/FILO}}.

The main application window (see Figure \ref{FILO:interface}) is divided into three sections. The input section allows to select predefined test from dropdown menu or choose an ontology file with an .owx or .owl extension. The output section shows the results of the unification process, including discovered unifier and/or relevant message. The options panel provides controls for setting the log level, as well as buttons to save the log file (“Save log file”) and to display diagnostic data (“Show statistics”).

\begin{figure}[H]
\centering
\includegraphics[width=0.6\textwidth]{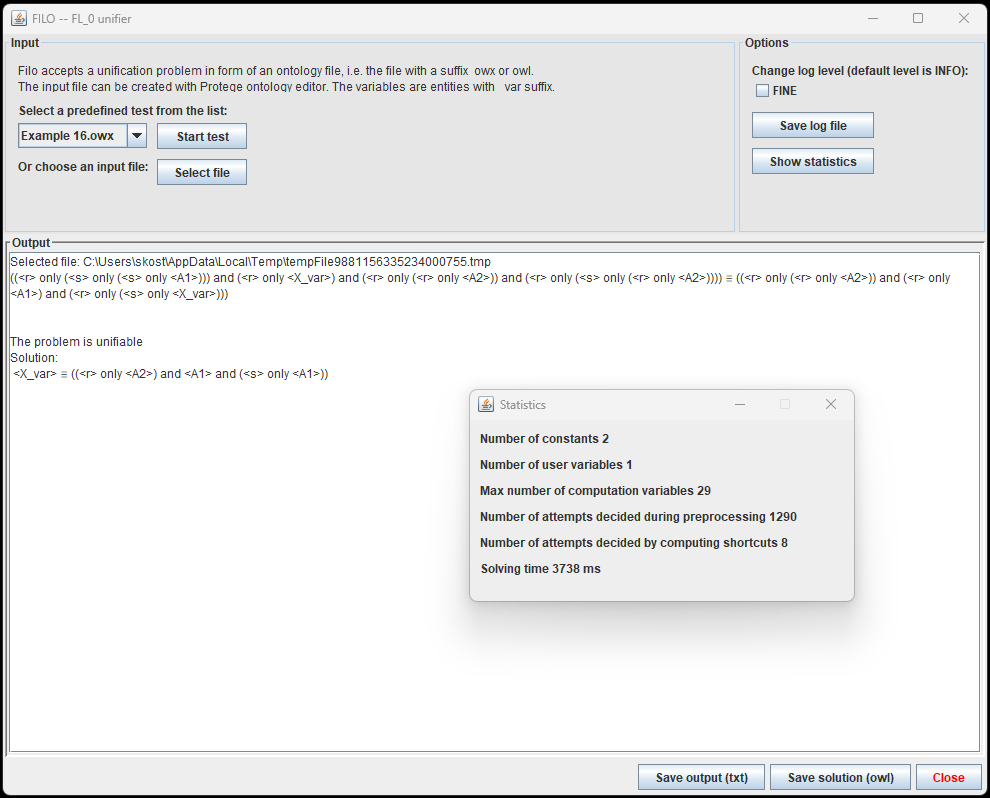}
\caption{FILO interface with statistics window.}
\label{FILO:interface}
\end{figure}

\subsection{Input}

A unification problem which can be solved by FILO must be given in form of an ontology file. Such an ontology may be created  using the ontology editor, Prot\'eg\'e\footnote{\url{https://protege.stanford.edu/}}. The concept names in such an ontology are treated by FILO
as constants, unless they have a suffix "\_var". The goal subsumptions are defined as general class axioms or concept subsumptions defined as class hierarchy statements. FILO recognizes concept names (constants and variables),
intersection of concepts and value restrictions. It will report error if an existential restriction or
any other constructor is encountered during reading the input file.


Having ready the input file with the suffix \emph{owl} or \emph{owx}, a user can open it in FILO by pressing
\emph{Select file} button and navigating to a proper folder, choosing the file in the provided file chooser.
FILO provides numerous examples, which can be examined by choosing them from a combo box on the main panel and
pressing \emph{Start test} button.

	
	\subsection{Reading input -- flattening I}
	FILO uses OWL API to read the input file and define its own internal representation of the problem in form of a filo-model. Filo-model is a set of equivalences, subsumptions and definitions i a flat normal form.
	The flat here means that each value restriction is distributed over the conjuncts of its argument, and a  value restriction nested below the top level of a concept intersection is abstracted by an introduction of a system variable and a new definition (equivalence).
	
	\begin{example}\label{example:flatteningI}
		Let an input subsumption in the notation of \flo be:\\
		$\forall r(\forall r(A \sqcap B) \sqcap C) \sqcap \forall s.(A \sqcap C) \sqsubseteq A \sqcap \forall s.(\forall r(A \sqcap C) \sqcap B)$.
		
		After the fattening I stage we get the following set of subsumptions and equivalences:\\
		$\begin{array}{l}
		\forall r.Var3 \sqcap \forall r.C \sqcap \forall r.Var0 \sqcap \forall r.Var1 \sqcap \forall r.Var2 \sqsubseteq \forall s.B \sqcap A \sqcap \forall s.Var4 \sqcap \forall s.Var5,\\
		Var3 \equiv \forall s. A,\\
		Var1\equiv\forall r. B,\\
		Var4\equiv\forall r.A,\\		
		Var2\equiv \forall s.C,\\
		Var5 \equiv\forall r.C,\\
		Var0 \equiv \forall r.A.
		\end{array}$
	\end{example}

The variables: $Var0, Var1, Var2, Var3, Var4, Var5$ are \emph{system variables}, used to abstract subconcepts from the original subsumption.
The concepts in the problem after the flattening I are in normal form.
The following lemma follows from the properties of subsumption in \flo.
\begin{lemma} \label{lemma:flatteningone}
	Let $\Gamma$ be a \flo unification problem and $\Gamma'$ is an \flo unification problem obtained from
	$\Gamma$ by flattening I. Then $\Gamma$ is unifiable iff $\Gamma'$ is unifiable.
	\end{lemma}
	
A unification problem obtained from the input by flattening I is called a \emph{Filo model}.
	
	\subsection{Main loop}

In view of Lemma~\ref{lemma:one-constant}, we can solve a unification problem for each constant separately and
then combine solutions into one using intersection constructor.

Having created a filo-model, FILO starts a loop with respect to each constant.
Hence it enters the loop: \emph{for each constnant}. If there are no constants, only variables, then
the problem has a solution sending all variables to $\top$.

FILO attempts to solve the problem for each constant separately, and if it fails for any one of them,
it breaks the main loop immediately returning \emph{failure}. The problem is not unifiable.

\subsection{Creating generic goal -- flattening II}
For a given  constant FILO creats a \emph{generic goal} for a given constant $A$.
First all subsumptions and equivalences are transformed into subsumptions with one particle on the right side,
where all concepts are in normal form.

\begin{example}
	Consider Example~\ref{example:flatteningI}. At this stage the subsumption:
	$\forall r.Var3 \sqcap \forall r.C \sqcap \forall r.Var0 \sqcap \forall r.Var1 \sqcap \forall r.Var2 \sqsubseteq \forall s.B \sqcap A \sqcap \forall s.Var4 \sqcap \forall s.Var5$ is divided into:
	
	$\begin{array}{l}
	\forall r.Var3 \sqcap \forall r.C \sqcap \forall r.Var0 \sqcap \forall r.Var1 \sqcap \forall r.Var2 \sqsubseteq \forall s.B \\
	\forall r.Var3 \sqcap \forall r.C \sqcap \forall r.Var0 \sqcap \forall r.Var1 \sqcap \forall r.Var2 \sqsubseteq  A 	\\
	\forall r.Var3 \sqcap \forall r.C \sqcap \forall r.Var0 \sqcap \forall r.Var1 \sqcap \forall r.Var2 \sqsubseteq \forall s.Var4\\
\forall r.Var3 \sqcap \forall r.C \sqcap \forall r.Var0 \sqcap \forall r.Var1 \sqcap \forall r.Var2 \sqsubseteq  \forall s.Var5\\
		\end{array}$
		
		The remaining equivalences are transformed into subsumptions:\\
		$\begin{array}{ll}
				Var3 \sqsubseteq \forall s. A, & \forall s.A \sqsubseteq Var3,\\
			Var1\sqsubseteq\forall r. B, & \forall r.B \sqsubseteq Var1,\\
			Var4\sqsubseteq\forall r.A, & \forall r.A \sqsubseteq Var4,\\		
			Var2\sqsubseteq \forall s.C, & \forall s.C \sqsubseteq Var2\\
			Var5 \sqsubseteq\forall r.C, & \forall r.C \sqsubseteq Var5,\\
			Var0 \sqsubseteq \forall r.A & \forall r.A \sqsubseteq Var0
			\end{array}$
\end{example}

At the same time all constants that are not equal to a given constant $A$ are replaced by $\top$. FILO applies reduction here.
The subsumptions at this moment are still not flat.

\begin{definition}
	A subsumption of the form $C_1 \sqcap \cdots \sqcap C_n \sqsubseteq D$, where $C_1, \dots, C_n, D$ are particles, is \emph{not flat} if at least one of the following conditions obtains:
	\begin{enumerate}
		\item $D = \forall r.D'$, for a role name $r$,
		\item $D \not=\top$ and there is $i$, $1\le i \le n$ such that $C_i = \forall r.C$, for a role name $r$,
	\end{enumerate}
	\end{definition}
In order to flatten the subsumptions we use the  rules of Figure~\ref{figure:flattening}.

	\begin{figure}
\begin{framed}
			Consider a \textbf{non flat} subsumption in a problem $\Gamma$: $s=C_1 \sqcap \cdots \sqcap C_n \sqsubseteq P$, where $C_1, \dots, C_n, P$ are particles.
			(Notice that $\top$ does not occur in $s$.)
			\begin{enumerate}
				\item\label{rule:-r} If $P$  of the form $\forall r.P'$,
				replace $s$ with $s^{-r}$ in $\Gamma$.
			
				
				\item\label{rule:A} If $P = A$, replace $s$ with $s^A$.
				
				\item\label{rule:fulldecomposition} If $P = X$. (There is $C_i$ of the form $\forall r.C'_i$ or $C_i$ is the constant $A$)
			Remove $s$ from $\Gamma$ and:
					\begin{enumerate}
						\item\label{rule-s-r} for each $r \in \roles$, add $s^{-r}$
						\item\label{rule:constantdecomposition} and add $C_1^{A} \sqcap \cdots \sqcap C_n^{A} \sqsubseteq^? X_A$ to $\Gamma$.
					\end{enumerate}	
			\end{enumerate}						
\end{framed}
\caption{Flattening of $\Gamma$}	
\label{figure:flattening}
\end{figure}


The notation used in Figure~\ref{figure:flattening} is explained as follows.
If $s = C_1 \sqcap \cdots \sqcap C_n \sqsubseteq P$, then $s^{-r} = C_1^{-r} \sqcap \cdots \sqcap C_n^{-r} \sqsubseteq P^{-r}$.

Now for a particle $E$, we define $E^{-r}$ as follows.

	$
E^{-r} = \begin{cases}
	E^{r} & \text{ if } E \text{ is a variable and }
	 E^r \text{ its \emph{decomposition variable} }\\
	E' & \text{ if } E = \forall r.E' \\
	\top & \text{ in all other cases } 
\end{cases}
$

In a similar way we define $s^A = C_1^A \sqcap \cdots \sqcap C_n^A \sqsubseteq P^A$, where each particle in this subsumption is defined as follows. For a particle $E$:

	$
E^{A} = \begin{cases}
	E & \text{ if } E \text{ is $A$ or a variable}\\
	\top & \text{ in all other cases } 
\end{cases}
$

The variable $X_A$ that occurs in the flattening rule~\ref{rule:constantdecomposition} is a
new variable defined for $X$, called a \emph{constant decomposition variable}.
Hence rules~\ref{rule:-r}, and \ref{rule:fulldecomposition} may introduce new variables:
decomposition variables and constant decomposition variables.
The intended meaning of these variables is as follows. Consider $\gamma$ a solution of a problem $\Gamma$.

\begin{enumerate}
	\item\label{property:decomposition} For each decomposition variable we want $\gamma(X^r)$ to be equal to $\{P \mid \forall r.P \in \gamma(X)\}$. If the set
	$ \{P \mid \forall r.P \in \gamma(X)\}$ is empty, then $\gamma(X^r)$ should be $\top$.
	Let us notice here that there is at most one decomposition variable for a given variable $X$ and a role name $r$. This means that if one was already created, we do not create another one, but reuse the old one.
	\item\label{property:constantDecomposition} In a similar way we treat a constant decomposition variable $X_A$. $\gamma(X_A)$ is either $\top$,
	or it is $A$. $\gamma(X_A) = A$ if and only if $A \in \gamma(X)$.
\end{enumerate}

We partly enforce the property~\ref{property:decomposition} by introducing additional subsumptions to the goal.
These are the so called  \emph{increasing subsumptions} of the form $X \sqsubseteq \forall r.X^r$ (at most one for each $X$ and $r$). These subsumptions are not subject for flattening. They are kept separately from the usual not flat subsumptions.
A solution $\gamma$ for an increasing subsumption $X \sqsubseteq \forall r.X^r$  ensures that 	$\gamma(X^r) \subseteq \{P \mid \forall r.P \in \gamma(X)\}$. The other direction of this subset relation is ensured by other means.
The other direction $\{P \mid \forall r.P \in \gamma(X)\} \subseteq \gamma(X^r)$ is called a \emph{decreasing rule} and cannot be expressed as a goal subsumption.

The property~\ref{property:constantDecomposition} is also enforced by other means in the next steps of the algorithm.

After applying the flattening rules of Figure~\ref{figure:flattening} exhaustively, we obtain a generic goal
$\Gamma_A$ that contains only:
\begin{itemize}
\item \emph{flat subsumptions} of the form $Y_1 \sqcap \cdots \sqcap Y_n \sqsubseteq^? X$, where all particles are variables or a constant $A$ and 
\item \emph{increasing subsumptions} of the form $X \sqsubseteq^? \forall r.X^r$, where $X^r$ is a \emph{decomposition variable} and $X$ is a \emph{parent} variable for $X^r$.
\end{itemize}

%

The point of creating a generic goal with respect to a given constant is that at this point 
all necessary variables are created and we do not need to create them any more.
Some of these variables are redundant. They will be \emph{deleted} by substituting them with $\top$.
Another alternative would be to create a goal directly by guessing the values of variables \emph{on the way}. But we have found it difficult to backtrack on such guesses. 

\begin{example}\label{example:genericgoal}
	Let our problem be:\\
	$\Gamma = \{ X \sqsubseteq^? \forall r.A,\, Y \sqcap \forall r.X \sqsubseteq^? X, \, X \sqsubseteq^? \forall r.Y\}$.
	
	\noindent
	Then flattening II yields the generic goal \wrt the constant $A$:\\
flat subsumptions: $ \{ X^r \sqsubseteq^? A,\, Y^r \sqcap X \sqsubseteq^? X^r,\,  Y  \sqsubseteq^? X_A,\, X^r \sqsubseteq^? Y  \}$,\\
increasing subsumptions: $\{ X \sqsubseteq^? \forall r.X^r, Y \sqsubseteq^? \forall r.Y^r\}$\\
	where $X^r, Y^r$ are decomposition variables and $X_A$ is a constant decomposition variable.
	\end{example}

Here we prove that the procedure of flattening II is correct.

\begin{lemma}(termination)  \label{lemma:fltwotermination}
	Let $\Gamma$ be a goal and $A$ a constant, an input for flattening II. Then the process of flattening II terminates
	with $\Gamma_A$, a generic goal for a constant $A$.
	\end{lemma}

\begin{proof}
	Let us assume that all constants different than $A$ are replaced in $\Gamma$ with $\top$.
	The lemma is a simple consequence of the following observations. The flattening rules are triggered by
	particles of the form $\forall r.P$ or  the constant $A$, occurring in the not flat subsumptions.
	There are only finitely many of such particles in $\Gamma$. 
	Each of the rules either removes a not flat subsumptions from $\Gamma$, or replaces such a subsumption
	with a new one with a strictly smaller number of the offending particles.
	
	No rule introduces any particles of the form  $\forall r.P$ or $B$ where $B$ is a constant, hence the process has to terminate, after taking $n$ steps, where $n$ is the number smaller or equal to
	the number of the particles of the form  $\forall r.P$ or $B$ where $B$ is a constant, in the original goal.	
\end{proof}

Next we show that flattening II preserves unifiability.

\begin{lemma}(completeness)  \label{lemma:fltwocompleteness}
	Let $\Gamma$ be a unification problem and $\gamma$ its solution.
	Let $\Gamma_A$ be a unification problem obtained from $\Gamma$ by flattening II with respect to a constant $A$.
	Then there is a substitution $\gamma'$, which is an extension of $\gamma$ for new variables,
	obeying the decreasing rule, such that $\gamma'$ is a solution for $\Gamma_A$.
\end{lemma}

\begin{proof}
	Consider rules from Figure~\ref{figure:flattening}. If no rule is applicable to any subsumption in $\Gamma$,
	this means that all subsumptions in $\Gamma$ are flat. 
	Since all subsumptions are flat, they are of the form:
	$Y_1 \sqcap \cdots \sqcap Y_n \sqsubseteq X$, there $Y_1, \dots, Y_n, X$ are variables or $A$.
	The set of increasing subsumptions is empty, and $\Gamma = \Gamma_A$.
	
Now, assume that a rule is applicable to a subsumption $s = C_1 \sqcap \cdots \sqcap C_n \sqsubseteq P$ in $\Gamma$.	
\begin{enumerate}
	\item If Rule~\ref{rule:-r} is applied, $P = \forall r.P'$ and we replace $s$ with $s^{-r}$ in $\Gamma$.
	The application may produce new decomposition variables and the increasing subsumptions.
	For each decomposition variable $X^r$, we extend $\gamma$ with the assignment:
	$X^r \mapsto \{P \mid  \forall r.P \in \gamma(X)\}$. Let $\gamma'$ be such an extension of $\gamma$.
	Due to the properties of subsumption in \flo, the extended $\gamma'$ unifies $s^{-r}$.
	It also unifies the increasing subsumption: $X \sqsubseteq \forall r.X^r$, and obeys the decreasing rule.
	Since it is an extension of $\gamma$ it unifies all the remaining subsumptions in the modified $\Gamma$.
	
	
	\item If Rule~\ref{rule:A} is applied, $P=A$, but the subsumption is not flat, hence there is a \emph{not flat} particle on its left side. The subsumption $s$ is replaced by $s^A$, where the non-flat particles are deleted. By the properties of the subsumption in \flo all such particles are redundant for the subsumption to hold. Hence $\gamma$ is still a unifier of the modified $\Gamma$.
	
	\item If Rule~\ref{rule:fulldecomposition} is applied, $P= X$ is a variable, but the subsumption is not
	flat, hence there is a \emph{not flat} particle on its left side.
	We remove this source of non-flattness by replacing $s$ with a set of subsumptions.
	For each $r \in \roles$, we add $s^{-r}$. This may introduce some new decomposition variables and increasing subsumptions.
	We extend $\gamma$ to these new variables as in the application of Rule~\ref{rule:-r}. This extension
	is a unifier of the new subsumptions and the increasing subsumptions too.
	Moreover, we extend $\gamma$ for the new constant decomposition variable $X_A$,
	according to its intended meaning: $\gamma(X_A) = A$ if $A \in \gamma(X)$ and $\gamma(X_A) = \top$ otherwise. Since $\gamma$ unified $s$, the extended $\gamma$ will unify
	$C_1^{A} \sqcap \cdots \sqcap C_n^{A} \sqsubseteq^? X_A$ which is added to $\Gamma$.
	
\end{enumerate}

Since flattening II terminates we can use induction on the steps already taken to be sure that
we finally obtain $\Gamma_A$ and a unifier, as required by the lemma.
\end{proof}

Finally we show soundness of flattening II.

\begin{lemma}(soundness) \label{lemma:fltwosoundness}
	Assume that all occurrences of constants in $\Gamma$ except for a given constant $A$ are replaced with $\top$.
	Let $\Gamma_{A}$ be the goal obtained from $\Gamma$ by exhaustive application of  flattening rules.
	Let $\gamma$ be a unifier of $\Gamma_A$. Then $\gamma$ is a unifier of $\Gamma$.	
\end{lemma}

\begin{proof}
	The proof is by induction on the number of flattening steps needed to obtain $\Gamma_{A}$.
	
	Hence we proof a slightly stronger statement: if $\Gamma'$ is obtained from $\Gamma$ by a number of flattening steps,
	then the unifier of $\Gamma'$ is also a unifier of $\Gamma$.
	
	If the number is $0$, then $\Gamma'=\Gamma$, and we are done.
	
	Let $\Gamma'$ be obtained in $k$ steps from $\Gamma$.
	Consider the last rule applied to obtain this goal. Hence there is a goal $\Gamma''$ which is 
	transformed to $\Gamma'$ by a rule application and $\Gamma''$ is obtained from $\Gamma$ by $k-1$ steps. By assumption $\gamma$ unifies $\Gamma'$. An inspection of the rules shows that
	$\gamma$ unifies $\Gamma''$. 
	
	Since $\Gamma''$ is obtained by $k-1$ steps from $\Gamma$ we can use  induction to state that $\gamma$ unifies $\Gamma$.
\end{proof}

From Lemma~\ref{lemma:one-constant} we see that if we obtain $\gamma_{A_1}, \dots, \gamma_{A_k}$ unifiers of generic goals defined for all
different constants $A_1, \dots, A_k$, then $\gamma = \gamma_{A_1} \cup \cdots \cup \gamma_{A_k}$ is a unifier of the original goal.
\subsection{Choice}
At this stage FILO has to guess which variables should be $\top$ or which should contain
a constant, or which do not contain a constant but are not $\top$ either,  under a solution.
Hence for each variable we have 3 choices: TOP, CONSTANT, NOTHING (which mean not $\top$ and no constant, but a value restriction is possible for such a variable).

FILO has to check if there is a possible unifier for each such choice. This would create a huge searching space
for even small number of variables. 
In order to restrict the space, we identify the variables for which only binary choice between the choice values makes sense.
Hence instead of keeping all choices for variables in one table, we have 3 tables to encode the current choice:
\emph{choiceTable} for variables with ternary possible choices, \emph{fixedChoice} for variables that have to have one fixed choice, \emph{binaryChoiceTable} for variables with binary choice only.
In order to compute next choice FILO keeps a given choice for normal variables and changes in a lexicographic order the choices for the
variables in the binary choice table. Then if all binary choices are exhausted it changes to the next choice
of ternary choices for variables in the choiceTable.

In order to \emph{fix} choice in this way, FILO applies the following rules for all subsumption $s$ and variables in the generic goal.
\begin{itemize}
	\item If $s$ has the left side empty, and the right side is a variable, then the choice for this variable is
	fixed to be TOP. (The id of this variable is added to the \emph{fixedChoice} table with the value TOP.)
	\item If $s$ has only constant on its left side and a variable $X$ on its right side,
	then $X$ is added to the binary choice table, with possible values TOP and CONSTANT.
	\item If the right side of $s$ has only constant and it has only one variable $X$ on the left side,
	then $X$ is fixed to be CONSTANT.
	\item If a parent variable $X$ is fixed to be TOP, then all its decomposition variables and 
	constant decomposition variable must be fixed to be TOP.
	\item If a parent variable $X$ is fixed and is not CONSTANT, then the constant decomposition variable 
	$X_A$ is fixed to be TOP.
	\item If a constant decomposition variable $X_A$ is fixed to be TOP, the parent variable $X$ cannot be CONSTANT. ($X$ is a binary choice variable in this case.)
	
\end{itemize}

After fixing choices, FILO computes current Choice for all variables. In this choice all values for 
ternary and binary variables are fixed. Now, we check the \emph{consistency} of this choice, by
looking again at the relations between the values of parent variables and its decomposition and constant decomposition variables. 
FILO uses the following rules.

\begin{itemize}
	\item The current choice is not consistent if a parent variable is TOP and its decomposition variable is not TOP.
	\item The current choice is not consistent if a constant decomposition variable is CONSTANT and 
	its parent variable is not CONSTANT.
	\item The current choice is not consistent if a constant decomposition variable is not CONSTANT and
	its parent variable is CONSTANT.
\end{itemize}

If a choice is not consistent FILO rejects it and the next choice is computed.

%

If there are no consistent choices for variables, FILO returns failure.

\begin{example}\label{example:choice}
In Example~\ref{example:genericgoal}, we have 5 variables: $X, X^r, X_A,  Y, Y^r$. The first choice is a 5-tuple: $(0,1,0,0,0)$, where $0$ signifies TOP and $1$ signifies CONSTANT.
Notice that $X^r$ is fixed to be CONSTANT, because of the first  subsumption.
This choice is not consistent since  $X$ is TOP but $X^r$ is not TOP.
A consistent choice for this generic goal is: $(1, 1 ,1 , 0,0)$, but the goal defined for this choice is not unifiable, because it forces $X_A$ to be CONSTANT but $Y$ is TOP, hence
$Y \sqsubseteq^? X_A$ cannot be unified. This decision is made by Implicit Solver (cf. the next subsection).
\end{example}

Let us notice here that the fixing choices and checking consistency of choices does not cover all
possible restrictions, and may be further improved.
We can also observe that the intended meaning of constant decomposition variables is secured by 
a consistent choice: a solution for $X$ should contain the constant iff the solution for its constant decomposition variable contains the constant.

\subsection{Goal}
When FILO finds a consistent choice for variables, it tries to create a goal with respect to this choice.
The goal is a set of \emph{unsolved} flat subsumptions, increasing subsumptions and \emph{start subsumptions}.
The start subsumptions are of the form $X \sqsubseteq^? A$, and they are created for all variables for which
choice assigns CONSTANT.

If choice is applied, many flat subsumptions become either false or trivially true.
This is checked by \emph{Implicit Solver}, a class that contains checks for subsumptions \wrt the choice.

\begin{figure}[H]
	\fbox{
		\begin{minipage}{\linewidth}
			\textbf{Implicit Solver rules:}	
			
			
			\begin{enumerate}
				\item\label{irules:1} If a variable $X$ occurs on the right side of $s$ and the choice for $X$ is TOP, then label $s$ as solved.				
				\item\label{irules:2} If a variable $X$ occurs on the left side of $s$ and the choice for $X$ is TOP, then delete
				this particle from $s$.
				\item\label{irules:3} If a particle $P$ (here $P$ is a variable or constant) occurs on the right side of $s$ and also on the left side (at the top level), then label $s$ as solved.
				\item\label{irules:4} If the constant $A$ occurs on the right side of $s$ and there is a variable $X$ on the left side, such that choice for $X$ is CONSTANT, then label $s$ as solved.
				\item\label{irules:5} If the constant $A$ occurs on the right side of $s$, but on the left side of $s$ there is neither  $A$ nor a variable $X$ with choice CONSTANT, then \textbf{fail}.
				\item\label{irules:6} If $A$ occurs on the left side of a subsumption $s$ and $X$ is on its right side, where choice for $X$ is not CONSTANT,
				then delete this occurrence of $A$ from $s$. 
				\item\label{irules:7} If choice for $X$ is CONSTANT  and $X$ occurs on the right side of $s$, but neither $A$ occurs on the left hand side nor a variable $Y$ with choice for $Y$ CONSTANT occurs on the left hand side of $s$, then \textbf{fail}.
				\item\label{irules:8} If $\top \sqsubseteq^? X$ is an unsolved subsumption and $X$ is not TOP, then \textbf{fail}. If $A \sqsubseteq^? X$ is an unsolved subsumption and either $X$ is NOTHING or there is a decomposition variable $X^{r}$ which is not TOP, then \textbf{fail}.
				\item\label{irules:9} If there are no unsolved subsumptions left, return \textbf{success}.
			\end{enumerate}
		\end{minipage}
	}
	\caption{Implicit Solver}
	\label{figure:implicit}
\end{figure}

Let us notice that the rules of Implicit Solver although are based on the current choice,
are not checking consistency of choice values, but rather change the set of subsumptions.

The checks 5, 7 and 8 are \emph{critical}, because they can return fail for the current choice. They are
performed before other checks.

%

If there are no more unsolved subsumptions, then the problem has a solution which is obtained from the choice and the increasing subsumptions. FILO returns \emph{success} and constructs a solution.

\begin{example}\label{example:goal}
	Continuing with Example~\ref{example:choice}, FILO searches for a consistent choice, and finds $(1,1, 1,  1, 0)$. 
	For this choice $X^r \sqsubseteq^? A$ is trivially satisfied. 
	We are left with  unsolved flat subsumptions: $Y^r \sqcap X \sqsubseteq^? X^r, \, Y \sqsubseteq^? X_A,\, X^r \sqsubseteq^? Y$.
	Moreover, we have 4 start subsumptions: $X \sqsubseteq^? A, X^r \sqsubseteq^? A, X_A \sqsubseteq^? A, Y \sqsubseteq^? A$.
\end{example}

We say that a unifier $\gamma$ \emph{conforms} to a given choice $C$ iff for every variable $X$: $\gamma(X) = \top$ iff the choice for $X$ is TOP, $A \in \gamma(X)$ iff the choice for $X$ is CONSTANT, and both $A \not\in \gamma(X)$ and $\gamma(X) \not= \top$ iff the choice for $X$ is NOTHING.

The process of constructing a \emph{goal} $\Gamma$ from a generic goal  $\Gamma_A$ is terminating and correct \wrt unifiability.
This is formulated in the following two lemmas.

\begin{lemma}(completeness) \label{lemma:implcompleteness}
	Let $\gamma$ be a unifier of $\Gamma_A$. Then there is a choice $C$ and a goal $\Gamma$ obtained from $\Gamma_A$ 
	for which
	$\gamma$ conforms to $C$ and $\gamma$ unifies $\Gamma$.
\end{lemma}

\begin{proof}
	For each variable $X$ occurring in $\Gamma_A$, let us define a choice $C$ as follows:
	\begin{itemize}
		\item TOP iff $\gamma(X) = \top$,
		\item CONSTANT iff $A \in \gamma(X)$,
		\item NOTHING iff $A \not\in \gamma(X)$ and $\gamma(X) \not= \top$.
	\end{itemize}
	Obviously, 
	$\gamma$ conforms to $C$.
	
	The proof is by induction on the number of steps to obtain $\Gamma$. We will show that if $\Gamma'$ is obtained from $\Gamma_A$ by a number of steps, then the unifier of $\Gamma_A$ is also a unifier of $\Gamma'$.
	
	If the number is $0$, then $\Gamma'=\Gamma_A$, then the lemma is satisfied by assumption.
	
	Let $\Gamma'$ be obtained in $k+1$ steps from $\Gamma_A$. Hence there is a goal $\Gamma''$ such that $\Gamma''$ is obtained from $\Gamma$ in $k$-steps and 
	$\Gamma'$ is obtained from $\Gamma''$ by an application of one rule from Figure~\ref{figure:implicit}.  By assumption $\gamma$ unifies $\Gamma''$.
	
	Consider this rule. It must be one of the rules 1 through 8, because Rule 9 does not change a goal. Hence one of these rules is applied to a subsumption $s = C_1 \sqcap ... \sqcap C_n \sqsubseteq P$ in $\Gamma''$.  An inspection of the rules shows that $\gamma$ unifies $\Gamma'$.
	
	\begin{enumerate}
		\item If one of Rules 1, 3, 4 is applied, $s$ is deleted from the set of flat subsumptions. Hence since $\gamma$ is a unifier of $\Gamma''$, it is also a unifier of  $\Gamma'$. 
		\item If Rule 2 is applied, we delete a particle $C_i$ from the left side of $s$ such that $\gamma(C_i) = \top$. Such $C_i$, occurring on the left side of $s$, is redundant for the  subsumption to hold. Hence after this deletion, $\gamma$ is a unifier of $\Gamma'$.
		\item 
		Since $\gamma$ is a unifier of $\Gamma''$, Rules 5, 7 and 8 cannot be applied.
		\item If Rule 6 is applied, it means that $A \not\in \gamma(X)$ and there is $C_i$ such that $C_i = A$. Then, such $C_i$, occurring on the left side of $s$, is redundant for the modified subsumption to hold. Hence $\gamma$ is a unifier of $\Gamma'$ after the deletion of $C_i$.
	\end{enumerate} 
\end{proof}

\begin{lemma}(soundness) \label{lemma:implsoundness}
	We have two cases to show.
	\begin{enumerate}
	\item Let $\Gamma_A$ be a generic goal defined for a given constant $A$ and let $\Gamma$ be a goal defined from $\Gamma_A$ based on a
	 choice $C$. If there is a unifier of $\Gamma$ conforming to the choice $C$ then this unifier is also a unifier of $\Gamma_A$.
	 \item Let  $\Gamma_A$ be a generic goal defined for a given constant $A$, $C$ a consistent choice for variables such that Rule 9 of Figure~\ref{figure:implicit} applies. (Hence no $\Gamma$ is created). Then there is a unifier $\gamma_A$ of $\Gamma_A$.
	 \end{enumerate}
\end{lemma}

\begin{proof}
	\begin{enumerate}
\item 	We first prove the first statement.
	Assume that $\gamma$ is a unifier of $\Gamma$ which conforms to a choice $C$. The proof is by induction on the number of steps to obtain $\Gamma$. We will show that if $\Gamma'$ is obtained from $\Gamma_A$ by a number of steps, then the unifier of $\Gamma'$ is also a unifier of $\Gamma_A$.
	
	If the number is $0$, then $\Gamma'=\Gamma_A$, and the lemma is true by assumption.
	
	Let $\Gamma'$ be obtained in $k+1$ steps from $\Gamma_A$.
	
	Consider the last rule applied to obtain this goal. Hence there is a goal $\Gamma''$ which is 
	transformed to $\Gamma'$ by a rule application and $\Gamma''$ is obtained from $\Gamma_A$ by $k$ steps. By assumption $\gamma$ unifies $\Gamma''$. An inspection of the rules shows that
	$\gamma$ unifies $\Gamma'$. 
	Note that we do not have to inspect failing checks or Rule 9 (this rule does not change anything).
	
	\begin{enumerate}
		\item If Rule 1 was applied, it means that a subsumption $s= C_1 \sqcap ... C_m \sqsubseteq X$ was deleted, where choice for $X$ was TOP. By conformity of $\gamma$, $\gamma(X) = \top$. So, $\gamma$ is a unifier of $s$ and thus of $\Gamma''$.
		
		\item If Rule 2 was applied, it deleted a particle $C_i=X$ from a subsumption which was of the form: $s = C_i \sqcap ... C_m \sqsubseteq P$, where the choice for $X$ was TOP. By conformity of $\gamma$, we know that $\gamma(X) = \top$. Then, $X$ is redundant for the subsumption $s$ to hold. Hence, $\gamma$ is a unifier of $s$ before the deletion and thus it is a unifier of $\Gamma''$.
		
		\item If Rule 3 was applied, then it means that a subsumption $s= P \sqcap C_1 \sqcap ... C_m \sqsubseteq P$ was deleted. This subsumption is true under any substitution hence $\gamma$ is a unifier of $s$ and thus of $\Gamma''$.

		\item If Rule 4 was applied, then it means that it means that a subsumption $s= C_1 \sqcap ... C_m \sqsubseteq A$ was deleted, where $C_i = X$ and 
		choice for $X$ was CONSTANT. By conformity of $\gamma$, $A \in \gamma(X)$. Then $\gamma$ is a unifier of $s$ and thus of $\Gamma''$.
		
		\item If Rule 6 was applied, a particle $A$ is deleted from a subsumption of the form: $s = A \sqcap C_i \sqcap ... C_m \sqcap X$, where choice for $X$ is not CONSTANT. By conformity of $\gamma$, we know that  $A \not\in \gamma(X)$. 
		Then, $A$ is redundant for the subsumption $s$ to hold. Hence, $\gamma$ is a unifier of $\Gamma''$.
	\end{enumerate}
	
	Since $\Gamma''$ is obtained by $k-1$ steps from $\Gamma_A$ and $\gamma$ unifies $\Gamma''$, we use  induction to state that $\gamma$ unifies $\Gamma_A$.
	
	\item Now we prove the second statement of the lemma.
If Implicit Solver applies Rule 9 of Figure~\ref{figure:implicit}, then all subsumptions are solved due to choice for variables. 
FILO constructs a unifier $\gamma_A$ conforming to the current choice using recursion.
\begin{itemize}
	\item If choice for $X$ is TOP, then $\gamma_A(X) := \top$
	\item If choice for $X$ is CONSTANT, then if $\gamma(X)$ is not yet defined, then
	$\gamma(X) := A$, otherwise $\gamma_A(X) = \gamma_A(X) \cup \{A\}$.
	\item If $X^r$ is a decomposition variable for $X$ and $\gamma_A(X^r)$ is defined and not $\top$,
	then if $\gamma_A(X)$ is not defined, then  $\gamma_A(X) = \{\forall r.P \mid P \in \gamma_A(X^r)\}$, otherwise  $\gamma_A(X) = \gamma_A(X) \cup \{\forall r.P \mid P \in \gamma(X^r)\}$.
\end{itemize}
Notice that the consistency of choice is crucial for this construction. For example, if choice for $X$ is TOP, then the choice for $X^r$ must also be TOP.
This construction terminates, because since flat subsumptions are solved by choice, they do not require creation of new particles. The first two steps of the construction assure the unification of flat subsumptions and the third step is needed to unify the increasing subsumptions.	
	\end{enumerate}
\end{proof}


\subsection{Computing with shortcuts}
All the previous steps of the procedure are either polynomial or nondeterministic polynomial. This
step is inherently exponential, although still the actual time needed depends on the problem.
We assume that FILO produced a goal $\Gamma$ from a generic goal for a constant $A$.
$\Gamma$ contains flat unsolved subsumptions, increasing subsumptions and start subsumptions.
This stage is triggered if the set of flat unsolved subsumptions and start subsumptions in the goal are not empty.

At this stage the algorithm works with \emph{shortcuts} \ie subsets of variables and constant from the goal.

\begin{definition}
	\begin{enumerate}
\item A set  $\{E_1, \dots, E_k\}$, where each $E_i$ is either a variable from  $\Gamma$ or a constant, \emph{satisfies}
a flat subsumption $C_1 \sqcap \cdots \sqcap C_n \sqsubseteq D$	if it satisfies the following implication:
if $D \in \{E_1, \dots, E_k\}$, then $\{E_1, \dots, E_k\} \cap \{C_1, \dots,C_n\} \not=\emptyset$
	
\item A set $\{C_1, \dots, C_n\}$, where each $C_i$ is either a variable from  $\Gamma$ or a constant, is 
a shortcut iff it \emph{satisfies} all unsolved flat subsumptions in $\Gamma$.
\end{enumerate}
\end{definition}
\begin{example}
	Let the unsolved flat subsumptions in $\Gamma$ be: $Y \sqcap Z \sqsubseteq X,\, Y \sqcap U \sqsubseteq X$.
	The set $\{X, Z \}$ is not a shortcut, because the second subsumption is not satisfied.
	The sets $\{X, Z, Y \}$, $\{Y, Z\}$, $\{X, Y\}$ are examples of shortcuts in this case.
\end{example}
%

We have to identify some relations between shortcuts.

\begin{definition}\label{definition:resolving}
	\begin{enumerate}
\item 	Let $S_1, S_2$ be two shortcuts. We say that $S_2$ resolves $S_1$ \wrt to a role name $r$ if the following conditions obtain:
	\begin{enumerate}
		\item\label{condition:increasing} There is a decomposition variable $X^r$ in $S_1$ and for each decomposition variable $Y^r \in S_1$, its parent $Y$ is in $S_2$.
		\item\label{condition:decreasing} If $X \in S_2$ and $X^r$ is defined, then $X^r \in S_1$.
	\end{enumerate}
	\item We call a shortcut $S$ \emph{resolved} in a set of shortcuts \S\ iff for each role name $r$ such that there is a decomposition variable 
	$X^r$ in $S$, there is $S' \in \S$ such that $S$ is resolved \wrt the role name $r$ with $S'$.
	\end{enumerate}
\end{definition}
\begin{example}
The following is the example of the resolving relation between the shortcuts.\\
	\begin{tikzpicture}
		\matrix (m) [matrix of math nodes, row sep=3em, column sep=4em, minimum width=2em]
		{ S_2 = \{X, Y^s, Z\} \\
				S_1 = \{U, Z^r, Y^{sr}\} \\};
		\path[-stealth]
		(m-2-1) edge node [left] {$r$} (m-1-1)
		;
	\end{tikzpicture}
	Here $X$ does not have $X^r$ defined, $U$ is not a decomposition variable $Y^s$ has a decomposition variable $Y^{sr}$ defined.
\end{example}

Here we make an observation, that shortcuts may be used to attempt to construct a unifier of $\Gamma$ using the resolving relation between them.
If for example, a new particle $P$ is assigned to all variables in a shortcut $S_1$ then we may be sure that this assignment makes all flat subsumptions in $\Gamma$ true. If $X^r$ is in $S_1$, then this assignment makes the increasing subsumption $X \sqsubseteq \forall r.X^r$ not true, because
$\forall r.P$ is not in an assignment for $X$ ($P$ is assumed to be new).
But since there is $S_2$, a shortcut that resolves $S_1$ \wrt the role name $r$, we can repair this situation, by putting 
the new particle $\forall r.P$ into the assignment for all variables in $S_2$. 
By the definition of the resolving relation, $X \in S_2$, hence $\forall r.P$ is in $X$.

One can also notice that thus the condition~\ref{condition:increasing} of Definition~\ref{definition:resolving} allows us to build assignment for variables in such a way as to unify
the increasing subsumptions, and the condition~\ref{condition:decreasing} makes sure that also the decreasing rule is satisfied.

FILO computes shortcuts in such a way, that together with the resolving relation they yield an acyclic graph.

At first FILO computes all shortcuts that do not contain decomposition variables. These are the so called \emph{shortcuts of height 0}. 

The set of shortcuts are then created recursively. If a decomposition variable has a parent that occurs in an already computed shortcut then it is added to the set of the so called \emph{good variables}. 
A new shortcut $S$ (possibly containing decomposition variables  but only those that are good) is added to 
the set of already computed shortcuts if it is resolved in the set of already computed shortcuts \S.
FILO greedily searches for a shortcut resolving $S$ in \S and stores the information in the structure
representing the shortcut $S$. This information is then used in creating a solution if there is one.

Computing shortcuts terminates if no more shortcuts are added to the set of already computed ones, or
if the so called \emph{initial} shortcut is computed.
The initial shortcut has the form $S_{ini} = \{X \mid X \sqsubseteq^? A \text{ is a start subsumption }\} \cup \{A\}$.

In the first case (no more shortcuts), FILO returns \emph{failure} for the current Choice for variables.

In the second case (initial shortcut is resolved) FILO returns \emph{success}. It then proceeds to construct an example of a unifier based
on the resolving relation between shortcuts. 

\begin{example}
	Continuing with Example~\ref{example:goal}, the initial shortcut is $\{A, X, Y, X^r, X_A \}$.
	$A$ can occur only in the initial shortcut.
	The only shortcut of height $0$ is $\{X\}$. FILO proceeds to compute shortcuts of height 1: $\{X, X^r\}$, $\{X, Y, X^r\}$,  $\{X, Y, X^r, X_A\}$ and finally
	 $\{X, Y, X^r, X_A, A\}$ which is also the initial shortcut. At this moment the computation is terminated with \emph{success}.
	 The solution is: $[X \mapsto A \sqcap \forall r.A,\, Y \mapsto A, X^r \mapsto A, X_A \mapsto A]$. FILO shows only user variables in the final display of the solution.
\end{example}

Finally, we can prove that the computing with shortcuts stage is correct with respect to the unifiability of the goal $\Gamma$ obtained
from a generic goal $\Gamma_A$ for a given constant $A$.

\begin{lemma}(completeness) \label{lemma:shcompleteness}
	If  the set of unsolved flat subsumptions of $\Gamma$ is not empty, the set of start subsumptions of $\Gamma$ is also not empty, and $\Gamma$ is unifiable, then computation of shortcuts terminates with \emph{success}.
\end{lemma}
\begin{proof}
	If the assumptions of the lemma are satisfied, then the computation of shortcuts is triggered.
	We show that FILO will not terminate with failure.
	
	Consider a unifier $\gamma$ of $\Gamma$. We assume that $\gamma$ does not introduce any constant other then $A$. 
	We also assume that it is \emph{minimal} with respect to the role depth of the particles it introduces in the assignments for variables.
	
	For every flat subsumption: $C_1 \sqcap \cdots \sqcap C_n \sqsubseteq D \in \Gamma$,
	$\gamma(C_1) \sqcap \cdots \sqcap \gamma(C_n )\sqsubseteq \gamma(D )$.
	We can assume that $\gamma$ assigns reduced concepts, hence each $\gamma(C_i)$ is a conjunction of particles.
	We choose the particles in the range of $\gamma$ of the maximal role depth.
	If $P$ is such a particle, it defines a shortcut of height $0$ in the following way:
	$S = \{X \mid P \in \gamma(X)\}$. By the properties of subsumption in \flo, this set must be a shortcut.
	It cannot contain decomposition variables, hence it is a shortcut of height $0$.
	
	Hence FILO will not terminate with failure while computing shortcuts of height $0$.
	
	Now assume that FILO did not terminate with failure at the $i$'th round of computation, and
	the initial shortcut is not yet computed. In the $i$'th round we have considered particles in the range of $\gamma$
	of the role depth $k$. (We have defined resolved shortcuts for such particles, just as in the case of height $0$ shortcuts.)
	
	Let us consider a resolved shortcut defined for a particle $P_k$ of the role depth $k$, $S_k = \{X \mid P_k \in \gamma(X)\}$.
	This shortcut is in the set \S\ of the already computed shortcuts.
	
	Now we extend this set for the shortcuts defined for particles of the role depth $k-1$.
	Indeed, $P_k$ must be either a constant or of the form $\forall r.P_{k-1}$. We consider these two cases separately.
	
	If $P_k = \forall r.P_{k-1}$, and $\forall r.P_{k-1} \in \gamma(X)$, where $X \in S_k$. By minimality of $\gamma$,
	there must be a variable $Y \in S_k$ such that $Y^r$ is defined. Otherwise replacing $P_k$ with $A$ would also be a unifier.
	(Remember that $S_k$ contains all variables with $P_k$ in $\gamma$.)
	
	$P_{k-1} \in \gamma(Y^r)$. Hence there must be a shortcut $S_{k-1} = \{X \mid P_{k-1} \in \gamma(X)\}$. This shortcut is resolved
	with $S_k$ \wrt the role name $r$. Hence FILO will not terminate with \emph{failure} at $i$'th round, but will compute $S_{k-1}$.
	
	Now consider the case where $P_k$ is a constant. Hence $P_k$ is really of the role depth $0$.
	Then $S_k$ is the initial shortcut and FILO terminates with \emph{success}.
	
\end{proof}

\begin{lemma}(soundness) \label{lemma:shsoundness}
	If FILO has computed shortcuts for $\Gamma$ and terminated with \emph{success}, then  $\Gamma$ is unifiable.
\end{lemma}
\begin{proof}
	Since FILO terminated with \emph{success}, it computed the initial shortcut: $\{X \mid X \sqsubseteq A\}$. FILO stores information about
	the resolving relation with each computed shortcut.
	Hence we can recursively construct a unifier.
	\begin{enumerate}
		\item At first we substitute all variables in the initial shortcut with $A$.
		\item For the recursion, assume that a particle $P$ was assigned to all variables in a shortcut $S$.
		\item If $S$ is resolved with $S'$ with respect to the role $r$, assign the particle $\forall r.P$ to all
		variables in $S'$.
	\end{enumerate}
	
	Since the resolving relation is acyclic, this process with terminate. Indeed it will terminate with the assignments for particles of maximal role depth
	to the variables in some shortcuts of height $0$.
	
	By the properties of shortcuts we are sure that this assignments unifies all flat subsumptions in $\Gamma$, all start subsumptions and 
	the increasing subsumptions. The decreasing rule is also satisfied, because of condition~\ref{condition:decreasing} in Definition~\ref{definition:resolving}.
\end{proof}
	
	\subsection{Complexity and correctness}
	
	Here we sum up the facts about the algorithm implemented in
	FILO and described here.
	
	\begin{theorem}(termination and complexity)
	The algorithm implemented in FILO terminates on every unification problem in \flo in at most exponential time in the size of the problem.
	The unification problem in \flo is ExpTime-complete.
	\end{theorem}
	\begin{proof}
		The algorithm terminates:
		\begin{enumerate}
			\item While reading a unification problem form a file, FILO preforms flattening I. It terminates, because there are finitely many
			sub-concepts of the form $\forall r.C$ where $C$ is a complex concept, in the problem. After distributing value restriction $\forall r.(C_1 \sqcap C_2)$ over
			conjunction $\forall r.C_1 \sqcap \forall r.C_2$ and  abstracting  sub-concepts that are not variables or constants with fresh variables (Example~\ref{example:flatteningI}) this stage must terminate with a structure
			called \emph{filo model}. This is done in polynomial time in the size of the problem. New variables are created.
			\item Next, if there are no constants, FILO terminates. Otherwise, for each constant FILO enters a loop. Since there are finitely many constants
			in a problem, FILO terminates if every such loop terminates. Hence consider a loop for one constant.
			\item  First flattening II is preformed, which yields a \emph{generic goal}.
			This stage terminates by Lemma~\ref{lemma:fltwotermination} This is a polynomial time step, where some new variables are created.
			\item Second, a consistent choice  for variables is searched for. If there is no consistent choice, then
			all choices are checked, and then FILO terminates with failure. This is a non-deterministic polynomial step. Hence
			failure at this stage may cost exponential time in the number of all variables.
			\item If a consistent choice is found, then Implicit Solver performs critical checks. It can fail for the current choice only.
			Otherwise FILO creates a \emph{goal}. If the goal does not have any flat subsumptions, FILO terminates with success.
			This is still a polynomial time step.
			\item Next, if there are unsolved flat subsumptions, FILO enters the phase of computing shortcuts.
			It can at most compute all possible shortcuts, hence exponentially many of them. If the initial shortcut is not computed,
			it fails for the current choice. 
			\item If the initial shortcut is computed FILO terminates with success.
		\end{enumerate}
		
		Overall a loop for one constant may take exponentially many steps (number of choices) times exponentially many steps (the number of shortcuts),
		which yields exponential time for one loop. There are polynomially many such loops possible, hence FILO takes at most exponential time in the size of the problem, to terminate.
		
	\end{proof}
	\begin{theorem}(soundness)
	If FILO terminates with \emph{success} on a unification problem $\Gamma$, then $\Gamma$ is unifiable.
	\end{theorem}
	
		\begin{proof}
		Let $\Gamma$ be a problem. Assume that FILO run on $\Gamma$ terminates with \emph{success}.
		\begin{enumerate}
		\item Flattening I on an input is sound by Lemma~\ref{lemma:flatteningone}.
		
		\item  If there are no constants in $\Gamma$, FILO answers that the problem is unifiable and returns $\{ X \mapsto \top \mid X \in \mathcal{C}_v \}$ as an example of a solution.
		Obviously, if a goal does not contain constants, then all particles are either of the form $\forall v.\top$ or $\forall v.X$. If all
		variables are mapped to $\top$, all concepts in the problem reduce to $\top$. Hence the problem is unified by this substitution.
		
		\item  There are constants and FILO answers that the problem is unifiable and returns $\gamma = \bigcup_{A \in \mathcal{C}_c} \gamma_A$ as an example of a solution, where each $\gamma_A$ is an example of a solution for a generic goal $\Gamma_A$ created for a particular constant $A$. 
		If $\gamma_A$ is a unifier of $\Gamma_A$, by Lemma~\ref{lemma:fltwosoundness} and Lemma~\ref{lemma:one-constant}, $\gamma$ is a unifier of $\Gamma$.
		
		Hence we consider soundness of FILO terminating a loop for a constant with success. If it does so, then $\Gamma_A$ is unifiable.
		
		There are two possibilities: FILO terminates the loop because of Rule 9 of Figure~\ref{figure:implicit} or 	
		 FILO terminates the loop with success after computing the initial shortcut for a goal. 
		In both cases by   Lemma~\ref{lemma:implsoundness} and Lemma~\ref{lemma:shsoundness}, we have that $\Gamma$ is unifiable.

		\end{enumerate} 
%
		\end{proof}
		
	\begin{theorem}(completeness)
		If a unification problem $\Gamma$ is unifiable,
		then FILO terminates with \emph{success}.
	\end{theorem}

	\begin{proof}
		Let $\Gamma$ be a unifiable problem.   We have to consider two cases.
		\begin{enumerate}
			\item There are no constants in $\Gamma$, then $\Gamma$ is unifiable and  FILO returns \emph{success} after flattening I.
			\item If there are constants, by Lemma~\ref{lemma:one-constant}, $\Gamma_A$ a generic goal obtained by flattening II for
			each constant $A$ is unifiable. Hence there is a choice for variables defined by a unifier $\gamma_A$ of $\Gamma_A$.
			FILO will detect this choice and this will be a consistent choice on base of which a goal will be defined.
			\begin{enumerate}
				\item If Implicit Solver is able to solve all flat subsumptions in $\Gamma_A$ by Lemma~\ref{lemma:implcompleteness}, FILO returns success.
				\item If $\Gamma_A$ has unsolved flat subsumptions, then FILO computes shortcuts. By Lemma~\ref{lemma:shcompleteness}, FILO will compute
			the initial shortcut and return success.
			\end{enumerate}
		\end{enumerate}
%
	\end{proof}
	
\subsection{Output}
The output is displayed in the main panel text area. If the answer to the unification problem is positive, the user can save the displayed message in a text file or a solution in
form of an ontology, in an owl-file. In the top-right sub-panel \emph{Options} the user can change the log level (before opening an input file) from INFO to FINE.
The FINE level contains very detailed messages used to develop or debug the program. The log file can be saved. Otherwise, it will be lost with the next run of FILO.

The statistics' window contains basic information about the computation. The maximal number of variables is the number of variables in a generic goal created during computation. Since several such generic goals may be created, only the maximal number is displayed. This number may indicate the difficulty of a unification problem.

The number of cases decided by pre-processing is the number of goals that were dismissed (or terminated with success) by Implicit Solver, and the number of cases decided by computing shortcuts,
is the number of goals for which FILO had to enter the the phase of computing shortcuts.

The time of computation is displayed in miliseconds.
	
Statistics are not defined if a user terminates the computation	before it is finished.
	
	\section{Examples}
	
	In the sets of examples provided by FILO (in the combo box), Example 16 is taken from \cite{BaNa-JSC01} (page 11), the first paper in which the unification in \flo was introduced and solved.

The example contains a unification problem in form of equivalence:
	$\forall r(A_1 \sqcap \forall r.A_2) \sqcap \forall rs.X_1 \equiv^? 
	\forall rs(\forall s.A_1 \sqcap \forall r.A_2) \sqcap \forall r.X_1 \sqcap \forall rr.A_2$.
	
	FILO computes the same solution that is contained in the paper and the computation takes 9310 ms.
One can see in the statistics window that during the computation FILO worked with at most 29 variables, deciding 1290 goals already in the preprocessing stage, and trying to solve the problem by computing shortcuts 8 times.

	Out of the examples provided, Example 8 is most difficult.
	It contains two constants and is not unifiable for either of them. Hence it terminates 
	with failure after checking the generic goal produced for one constant only.
	Nevertheless it takes 29360 ms to terminate. The maximal number of variables is 33 and 1525 goals are
	rejected in the preprocessing stage, while the computation with shortcuts has been triggered 12 times.
	
	\section{Implementation improvements}
	
	The fist version of FILO presented in this paper, can be improved in many ways. We mention few of them here.

	FILO currently has no means of checking whether a unification problem is ground or an instance of matching. Consequently, in such cases, it proceeds with the computation, treating the problem as if it was a full unification problem with variables. 
	Therefore, for the time being, FILO should not be used as a subsumption decider or a matching solver. 
	In the future we will equip FILO with the ability to pre-check the input problem and apply a polynomial time algorithm for these simpler tasks.
	
	The first flattening step will almost always introduce system variables. We plan to 
	perform only one flattening step, by extending flattening II procedure. We hope this will reduce the number
	of variables and thus improve runtime efficiency.
	
	FILO can also be improved with respect to the Implicit Solver cases. By careful analysis of potential input subsumptions, the computation ca be shortened. 
	
	Checking choices in the lexicographic order can be perhaps replaced by a smarter way of the evaluation. One can think about using a SAT solver for this task.
%
%
%
%
%
\bibliographystyle{plain}	
	\bibliography{FLbottom}
\end{document}